\documentclass[10pt]{article}
\usepackage{dcolumn}
\usepackage{bm}
\usepackage{verbatim}       

\usepackage{amsthm}
\usepackage{amssymb}

\usepackage{booktabs}

\usepackage{amstext}
\usepackage{amsmath}
\usepackage{amsfonts}
\usepackage{units}
\usepackage{mathrsfs}
\usepackage{cite}

\newcommand{\p}{\partial}

\newcommand{\dd}{{\rm d}}

\newcommand{\bd}{\begin{definition}}                
\newcommand{\ed}{\end{definition}}                  
\newcommand{\bc}{\begin{corollary}}                 
\newcommand{\ec}{\end{corollary}}                   
\newcommand{\bl}{\begin{lemma}}                     
\newcommand{\el}{\end{lemma}}                       
\newcommand{\bp}{\begin{proposition}}            
\newcommand{\ep}{\end{proposition}}                
\newcommand{\bere}{\begin{remark}}                  
\newcommand{\ere}{\end{remark}}                     

\newcommand{\bt}{\begin{theorem}}
\newcommand{\et}{\end{theorem}}

\newcommand{\be}{\begin{equation}}
\newcommand{\ee}{\end{equation}}

\newcommand{\bit}{\begin{itemize}}
\newcommand{\eit}{\end{itemize}}
\newtheorem{theorem}{Theorem}[section]
\newtheorem{corollary}[theorem]{Corollary}
\newtheorem{lemma}[theorem]{Lemma}
\newtheorem{proposition}[theorem]{Proposition}
\theoremstyle{definition}
\newtheorem{definition}[theorem]{Definition}
\theoremstyle{remark}
\newtheorem{remark}[theorem]{Remark}

\begin{document}

\title{Special coordinate systems in pseudo-Finsler geometry and
the equivalence principle}

\author{E. Minguzzi\thanks{
Dipartimento di Matematica e Informatica ``U. Dini'', Universit\`a
degli Studi di Firenze, Via S. Marta 3,  I-50139 Firenze, Italy.
 e-mail: ettore.minguzzi@unifi.it }}

\date{}
\maketitle

\begin{abstract}
\noindent Special coordinate systems are constructed in a neighborhood of a point or of a curve. Taylor expansions can then be easily inferred for the metric, the connection, or the Finsler Lagrangian in terms of curvature invariants. These coordinates circumvent the difficulties of the normal and Fermi coordinates in Finsler geometry, which in general are not sufficiently differentiable. They are obtained applying the usual constructions to  the pullback of a horizontally torsionless connection. The results so obtained are easily specialized to  the Berwald or Chern-Rund connections and have application in the study of the equivalence principle in Finslerian extensions of general relativity.
\end{abstract}

\section{Introduction}

Finslerian modifications of Einstein's gravity have received renewed attention quite recently \cite{bogoslovsky98,gibbons07,lammerzahl12,pfeifer12,chang13,minguzzi13c,minguzzi14c}, while the mathematical interest in Finsler geometry  never faded \cite{bejancu99,bryant02,matveev09b,sabau10,bucataru10,bucataru12}. In these theories the motion of a free falling particle is described by a geodesic, this concept  being defined through the notion of  spray \cite{shen01b}, and as such it makes no reference to other properties of the particle such as its mass or its composition. We might say that the weak equivalence principle is naturally satisfied in these theories.

Still one would like to show that any free falling observer looking at neighboring free falling particles observes them moving uniformly over straight lines, at least within some approximation. In order to accomplish
this result it is necessary to show that natural coordinates can be defined in a neighborhood of the observer, and that  free particles move indeed on straight lines according to those coordinates.

In pseudo-Riemannian geometry one uses normal coordinates in a neighborhood of  a point or Fermi(-Walker) coordinates in a neighborhood of  a curve, both being built using the exponential map.
Unfortunately, both procedures fail in pseudo-Finsler geometry unless the space is Berwald ($G^\alpha_{\ \beta \gamma \delta}=0$). In fact, in normal coordinates at $\bar x$ any geodesic passing through $\bar x$ reads $x^\mu(t)= n^\mu t$, which, recalling the geodesic equation for general sprays (the reader not familiar with the next expressions and  notations is referred to the next section for an introduction to our terminology)
\[
\ddot x^\alpha+2G^\alpha(x,\dot x)=0,
\]
gives for any $n\ne 0$,
\[
G^\alpha(n t,n )=0.
\]
Differentiating three times with respect to $n$, setting $n=v$ and letting $t\to 0$ gives $G^\alpha_{\ \beta \gamma \delta}(\bar x,v)=0$. In these mathematical steps we have tacitly assumed that the normal coordinate system is $C^5$, otherwise we could not to write the geodesic equation in this chart and  differentiate three times. In conclusion:

\begin{proposition}
If $C^5$ normal coordinates exist at a point $x$ then for every $v\in T_xM$, $G^\alpha_{\ \beta \gamma \delta}(x,v)=0$.
\end{proposition}

Actually, more refined results have been obtained \cite{busemann55} which show that the existence of $C^2$ normal coordinate systems cannot be assumed in general pseudo-Finsler  spaces.
It is therefore natural to ask whether adapted coordinate systems can be introduced which simplify the expressions of the connection coefficients, metric and Finsler Lagrangian without passing from the Finslerian exponential map (for a detailed study of the Finslerian exponential map the reader is referred to \cite{whitehead32,whitehead33,minguzzi13d}).

This problem has been studied long ago by Veblen and Thomas \cite{veblen23} and by Douglas \cite{douglas27}. It has also been reconsidered by Pfeifer \cite{pfeifer14} through an approach framed on the tangent bundle.  The main strategy was suggested to Douglas by Thomas \cite[Eq. (8.10)]{douglas27}.
Thomas argued that the
normal coordinates at $\bar x$ should also depend on a vector $\bar v\in T_{\bar x} M$. He
suggested to solve the equations
\begin{align}
\ddot x^\alpha+G^\alpha_{\beta \gamma}(x,s) \dot x^\beta \dot x^\gamma&=0, \label{jne}\\
\dot s^\alpha+G^\alpha_{\beta \gamma}(x,s) \dot x^\beta s^\gamma&=0, \label{okd}
\end{align}
over a  interval [0,1] with the initial conditions $x(0)=\bar x$, $\dot{x}(0)=y$, $s(0)=\bar v$. The previous equations would provide a map  $y \mapsto x(1)$. Then $\{y^\alpha\}$ such that $y=y^\alpha e_\alpha$, would be the normal coordinates of $x(1)$. Observe that $y \mapsto x(1)$ is the usual exponential map if the space is Berwald ($G^\alpha_{\beta \gamma}$ is independent of velocity), thus these normal coordinates are the usual normal coordinates if the spray is a connection. The fact that the map provides a local diffeomorphism follows from the implicit function theorem as for the usual exponential map.

Notice that since the contraction $G^\alpha_{\beta \gamma}(x,s) s^\gamma=N^\alpha_\beta(x,s)$, gives the non-linear connection, in modern terminology (\ref{okd}) is nothing but the equation $D_{\dot x} s=0$ stating that $s$ is horizontal over $x(t)$ while (\ref{jne}) establishes that $x(t)$ is a geodesic for the pullback connection $s^* \nabla^B$ where $\nabla^B$ is the Berwald Finsler connection and $s$ is a section from the path $x(t)$ to $TM$. Since the paths $x(t)$ obtained for different choices of $y$ cover a whole neighborhood of $\bar x$ one could construct a section $s$ over a whole neighborhood of $\bar x$ in such a way that $D s(\bar x)=0$.

\begin{remark} \label{jdd}
Actually other notable  normal coordinates analogous to Douglas-Thomas' but based on (\ref{jne})-(\ref{okd}) where $G^\alpha_{\beta \gamma}$ is replaced by  the Chern-Rund connection $\Gamma^\alpha_{\beta \gamma}$ could be constructed. Our analysis will comprise this case. In order to distinguish it from the traditional one we shall call the coordinates a la  Douglas-Thomas' $\bar v$-{\em Berwald normal coordinates} and those obtained replacing $G^\alpha_{\beta \gamma}\to \Gamma^\alpha_{\beta \gamma}$, $\bar v$-{\em Chern-Rund normal coordinates}.
\end{remark}


In this work we use an approach according to which $s$ is directly a section defined in a neighborhood of $\bar x$ (we do not use  a different section over every curve passing through $\bar x$) such that $D s=0$ at $\bar{x}$. We then build normal coordinates for $s^* \nabla$ where $\nabla$ is either the Berwald or the Chern-Rund connection or more generally, a horizontally  torsionless Finsler connection with  trivial vertical coefficients. Whenever $s$ is the section determined by the local construction of Douglas and Thomas outlined in the previous paragraphs, and $\nabla$ is the Berwald Finsler connection  we recover the normal coordinates constructed by these authors. Whenever instead $\nabla$ is the Chern-Rund connection we recover the Chern-Rund normal coordinates of Remark \ref{jdd}.

However, we shall not work with such a rigid choice of $s$, in fact the lowest order terms in the Taylor expansion of the metric will turn out to be independent on how we select ${s}$ as long as it satisfies $s(\bar x)=\bar v$ and $D s(\bar x)=0$.

Thus our local coordinate system is indeed on $M$, but the exponential map we use to build it is well defined and sufficiently differentiable since it is the exponential map of a usual (non-metric) torsionless connection  $s^* \nabla$. The interesting fact is that  this coordinate system  will retain much of the classical properties  of the normal coordinate system for what concerns its ability to simplify the expression of the connection coefficients (Prop.\ \ref{wfy}-\ref{dib}). Moreover, as mentioned, the derivatives of the metric at $\bar x$ and hence its Taylor expansion, at least for what concerns the first terms which we have calculated, turn out to be independent of the chosen section, thus the derivatives at $\bar x$ which we obtain are the same that would have been obtained using the Douglas-Thomas construction.
We notice that the expansion of the metric was not determined in these early investigations.

Through the use of this coordinate system we will be able to clarify the notion of local observer in Finsler gravity. In fact, we will obtain a general formula which expresses the apparent forces in the comoving frame (Eq.\ (\ref{qki})).

\section{Elements of pseudo-Finsler geometry}
The purpose of this section is mainly that of fixing  notation and terminology. As in Finsler geometry there are many different notations we shall give some key coordinate expression which might allow the reader to make fast correspondences with notations he might be used to. Of course the objects introduced below can be given coordinate-free formulations, for those the reader is referred to \cite{antonelli93,abate94,bao00,shen01,szilasi14,minguzzi14c}.

Let $M$ be a paracompact, Hausdorff, connected, $n+1$-dimensional manifold. Let $\{x^\mu\}$ denote a local chart on $M$ and let $\{ x^\mu,v^\nu\}$ be the induced local chart on $TM$.
The Finsler Lagrangian is a  function on the slit tangent bundle $\mathscr{L}\colon TM\backslash 0 \to \mathbb{R}$ positive homogeneous of degree two in the velocities, $\mathscr{L}(x,sv)=s^2 \mathscr{L}(x,v)$ for every $s>0$.  The metric is defined as the Hessian  of $\mathscr{L}$ with respect to the velocities
\[
g_{\mu \nu}(x,v)= \frac{\p^2 \mathscr{L}}{\p v^\mu \p v^\nu},
\]
and in index free notation will be also denoted with $g_v$ to stress the dependence on the velocity. This Finsler metric provides a map $g\colon TM\backslash 0 \to  T^*M \otimes T^*M$. The tensor $C_{\alpha \beta \gamma}(x,v)=\frac{1}{2}\,\frac{\p}{\p v^\gamma}g_{\alpha \beta}(x,v)$ is called Cartan torsion.

Lorentz-Finsler geometry is obtained whenever $g_v$ is Lorentzian, namely of signature $(-,+,\cdots,+)$.
The definition of Lorentz-Finsler manifold can be found in \cite{beem70}. We note that it is particularly convenient to work with a Lagrangian defined on the slit bundle $TM\backslash 0$ since the theory of Finsler connections traditionally has been developed on this space. For what concerns applications to Finsler gravity we shall tacitly assume  that the signature is Lorentzian, but for the other results the signature could be arbitrary.

Let us recall some elements on the geometry of pseudo-Finsler connections (the reader is referred to \cite{minguzzi14c}). The Finsler Lagrangian allows us to define the geodesics as the stationary points  of the functional $\int \mathscr{L}(x,\dot x)\dd t$. The Lagrange equations are of second order and it turns out that a good starting point for the introduction of the Finsler connections is the notion of {\em spray}.

We recall that a spray over $M$ can be characterized locally as a second order equation
\[
\ddot x^\alpha+2G^\alpha(x,\dot x)=0,
\]
where $G^\alpha$ is positive homogeneous of degree two $G(x,s v)=s^2 G(x,v)$ for every $s>0$. Let $E=TM\backslash 0$, and let $\pi_M\colon E\to M$ be the usual projection. This projection determines a vertical space $V_e E$ at every point $e\in E$. A non-linear connection is a splitting of the tangent space $TE=VE\otimes HE$ into vertical and horizontal bundles. A base for the horizontal space is given by
\[
\{\frac{\delta}{\delta x^\mu}\}, \qquad \frac{\delta}{\delta x^\mu}=\frac{\p}{\p x^\mu}-N^\nu_\mu(x,v) \frac{\p}{\p v^\nu},
\]
where the coefficients $N^\nu_\mu(x,v)$ define the non-linear connection and have suitable transformation properties under change of coordinates.
The curvature of the non-linear connection measures the non-holonomicity of the horizontal distribution
\[
\left[ \frac{\delta}{\delta x^\alpha}, \frac{\delta}{\delta x^\beta}\right]= -R^\mu_{\alpha \beta} \frac{\p}{\p v^\mu}, \qquad R^\mu_{\alpha \beta}(x,v)=\frac{\delta N^\mu_\beta}{\delta x^\alpha}-\frac{\delta N^\mu_\alpha}{\delta x^\beta}.
\]
Given a section $s\colon U\to E$, $U\subset M$, we can define a covariant derivative for the non-linear connection as
\[
 D_{\xi} s^\alpha=(\frac{\p s^\alpha}{\p x^\mu} +N^\alpha_\mu(x, s(x)) )\xi^\mu .
\]
The flipped derivative is instead
\[
{\tilde{D}_{\xi}} s^\alpha=\frac{\p s^\alpha}{\p x^\mu} \, \xi^\mu+N^\alpha_\mu(x, \xi) s^\mu .
\]
and although well defined {\em is not} a covariant derivative in the standard sense since it is non-linear in the derivative vector $\xi$. As a consequence, we cannot speak of curvature of the flipped derivative.\footnote{
Some authors call it covariant derivative \cite{shen01}, but in our opinion this term should be reserved to $D$.}
 Observe that if $X,Y\colon M\to TM$ are vector fields then
 \[
 {\tilde D_{X}} Y- D_Y X=[X,Y].
  \]
A geodesic is a curve $x(t)$ which satisfies $D_{\dot x} \dot{x}=0$ (note that it can also be written ${\tilde D_{\dot x}} \dot x=0$).
We shall only be interested in the non-linear connection determined by a spray as follows
\begin{equation} \label{dis}
N^\mu_\alpha=G^\mu_\alpha:=\p G^\mu/\p v^\alpha.
\end{equation}
 The geodesics of this non-linear connection coincide with the integral curves of the spray.

The spray comes from a Lagrangian if the geodesics of the spray are  the stationary points of the action functional $\int \mathscr{L}\dd t$, that is
\begin{align}
2 {G}^\alpha(x,v)&= g^{\alpha\delta}\left( \frac{\p^2\mathscr{L} }{\p x^\gamma \p v^\delta} \,v^\gamma -\frac{\p\mathscr{L} }{\p x^\delta } \right) \label{axo}\\
&=\frac{1}{2}\, g^{\alpha \delta} \left( \frac{\p}{\p x^\beta} \,g_{\delta \gamma}+\frac{\p}{\p x^\gamma} \, g_{\delta \beta}-\frac{\p}{\p x^\delta}\, g_{\beta \gamma}\right) v^\beta v^\gamma . \label{axu}
\end{align}
Some of the results which we shall obtain will be independent of the compatibility of the spray with a Finsler Lagrangian (observe that the Douglas-Thomas' normal coordinates construction does not make use of this structure).

In Finsler geometry one can further define the linear Finsler connection $\nabla$, namely splittings of the vertical bundle $\pi_E\colon VE\to E$, $E=TM\backslash 0$. The  Berwald, Cartan, Chern-Rund and Hashiguchi connections are of this type. They are referred as {\em notable} Finsler connections.  Although different, they are all compatible with the same  non-linear connection. In fact, the covariant derivative $X\to \nabla_X L$ of the Liouville vector field $L\colon E\to VE$, $L=v^\alpha\p/\p v^\alpha$, vanishes precisely over a n+1-dimensional distribution which determines a non-linear connection.  For all the notable connections this distribution is always the same and is determined by the spray as in (\ref{dis}).

Each Finsler connection $\nabla$ determines two covariant derivatives $\nabla^H$ and $\nabla^V$ respectively being obtained from $\nabla_{\check X}$ whenever $\check X$ is the horizontal (resp.\ vertical) lift of a vector $X\in TM$. In particular $\nabla^H$ is determined by local connection coefficients $H^\alpha_{\mu \nu}(x,v)$ which are related to those of the non-linear connection by $N^\alpha_\mu(x,v)=H^\alpha_{\mu \nu}(x,v)v^\nu$. We shall distinguish between the Berwald horizontal derivative $\nabla^{HB}$ and the Chern-Rund or Cartan horizontal derivative, denoted $\nabla^{HC}$.
The horizontal coefficients of the Berwald connection read
\[
H^\alpha_{\mu \nu}:=G^\alpha_{\mu \nu}:=\frac{\p}{\p v^\nu}\,G^\alpha_\mu.
\]
The further derivative $G^\alpha_{\mu \nu \beta}$ defines the Berwald curvature.
As for the Chern-Rund or Cartan connection, the condition $\nabla^{HC} g=0$  gives
\[
H^\alpha_{\beta \gamma}:=\Gamma_{\beta \gamma}^{\alpha}:=\frac{1}{2} g^{\alpha \sigma} \left( \frac{\delta}{\delta x^\beta} \,g_{\sigma \gamma}+\frac{\delta}{\delta x^\gamma} \, g_{\sigma \beta}-\frac{\delta}{\delta x^\sigma}\, g_{\beta \gamma}\right).
\]
The difference
\begin{equation} \label{lan}
L_{\beta \gamma}^{\alpha}=G_{\beta \gamma}^{\alpha}-\Gamma_{\beta \gamma}^{\alpha}
\end{equation}
is the Landsberg (Finsler) tensor. The tensor $L_{\alpha \beta \gamma}(x,v)=g_{\alpha \mu}(x,v) L^{\mu}_{\beta \gamma}(x,v)$ is  symmetric and  $L_{\alpha \beta \gamma}(x,v) v^\gamma=0$.

A property of the flipped derivative which is a consequence of the horizontal compatibility of the  Chern-Rund or Cartan connections with the metric is
\[
\tilde{D}_u g_u(X,Y)= g_u(\tilde{D}_u X,Y)+g_u(X,\tilde{D}_u Y),
\]
for every vector $u\in T_p M\backslash 0$  and fields $X,Y\colon M\to  TM$. The linearity of the map $X\mapsto \tilde D_u X$ implies that $\tilde{D}_u$ can be extended to one-forms and hence tensors in the usual way. Thus the previous identity is simply the statement
\begin{equation}
\tilde D_u g_u=0.
\end{equation}

The horizontal-horizontal curvature\footnote{I prefer to denote the curvatures of the Finsler connection with $R^{HH}$, $R^{VH}$ and $R^{VV}$, in place of $R$, $P$ and $S$, as done by some authors. Indeed, I wish to make the notation less ambiguous since the letter $R$ is also used for the curvature of the non-linear connection.} $R^{HH}$ of any Finsler connection is related to the curvature of the non-linear connections as follows (see e.g.\ \cite[Eq.\ (67)]{minguzzi14c})
\[
R^{HH} {}^\alpha_{\ \beta \mu \nu}(x,v) v^\beta=R^\alpha_{\mu \nu}(x,v) .
\]
If the non-linear connection comes from a Finsler Lagrangian then we set
 $R_{\alpha \beta \gamma}=g_{\alpha \mu} R^\mu_{\beta \gamma}$ and $R^\alpha_{\beta}=R^\alpha_{\beta \gamma} v^\gamma$. In this case we have $R_{[\alpha \beta \gamma]}=0$, for a proof see for instance \cite[Eq.\ (73) and (87)]{minguzzi14c}.

\section{The equivalence principle}

The weak equivalence principle is the statement according to which the trajectory of a body on a gravitational field depends only on its initial position and velocity, and is independent of its composition and structure. In mathematical terms it states that free fall is represented by  geodesics where these paths are defined through a spray.

Sometimes, one can find a statement according to which the weak equivalence principle implies that any observer in free fall looking at neighboring test particles would observe them move uniformly over straight lines (at least up to higher order terms linear in position (tidal forces) or quadratic in the velocities).

Let us show that this is indeed the case. Under a change of coordinates $\tilde{x}(x)$ the connection coefficients of a horizontal connection transform as
\[
\tilde H^\alpha_{\beta \gamma}(x,v) =\frac{\p \tilde x^\alpha}{\p x^\sigma } \frac{\p x^\delta}{\p \tilde x^\beta} \frac{\p x^\mu}{\p \tilde x^\gamma} \,H^\sigma_{\delta \mu}(x,v) +\frac{\p^2  x^\sigma}{\p \tilde x^\beta  \p  \tilde x^\gamma }  \frac{\p \tilde x^\alpha}{\p  x^\sigma}  .
\]
Thus we have the following little trick which has been used, for instance in \cite{minguzzi14c}, to infer indentities and check  long calculations.

\begin{proposition} \label{poi}
For any given $(\bar x,\bar v)\in E=TM\backslash 0$, and any chosen Finsler connection, it is always possible to find local coordinates near $\bar x$ such that $\p_0(\bar x)=\bar v$ and
\[
H^\alpha_{\beta \gamma}(\bar x,\bar v)=G^\alpha_\beta(\bar x,\bar v)=G^\alpha(\bar x,\bar v)=\frac{\p \mathscr{L}}{\p x^\alpha}(\bar x,\bar v)=0.
\]
\end{proposition}

Observe that one can either choose to get $G^\alpha_{\beta \gamma}(\bar x,\bar v)=0$ choosing the Berwald Finsler connection or $\Gamma^\alpha_{\beta \gamma}(\bar x,\bar v)=0$ choosing the Cartan Finsler connection.

\begin{proof}
Suppose that $H^\alpha_{\beta \gamma}(\bar x,\bar v)\ne 0$, in a system of coordinates for which $x^\nu(\bar{x})=0$.  The change of coordinates such that $\tilde{x}^\alpha- H^\alpha_{\beta \gamma}(\bar x,\bar v) \tilde x^\beta \tilde x^\gamma=x^\alpha$ is clearly locally invertible since $\p x^\alpha/\p \tilde x^\beta(\bar{x})=\delta^\alpha_\beta$, and accomplishes the first equation. The second equation follows from $G^\alpha_\beta=H^\alpha_{\beta \gamma} v^\gamma$, the third from (positive homogeneity) $G^\alpha=2 G^\alpha_\beta v^\beta$, and the latter from the fact that $\frac{\delta \mathscr{L}}{\delta x^\gamma}=0$, see e.g.\ \cite[Prop.\ 3.6,4.5]{minguzzi14c}. The equation $\p_0(\bar x)=\bar v$ is accomplished with a last linear change of coordinates.
\end{proof}
Observe that in order to check a tensorial equation on $E$ it is sufficient to check it in the special reference frame given by the previous proposition. By covariance it will then hold in any coordinate system (recall that the coordinate system on $TM$ is induced from that on $M$). This trick provides a drastic help in calculations.

The previous proposition already implies that for every $(\bar{x}, \bar{v})$, representing the motion of an observer at a certain event, there is a coordinate system for which $G^\alpha_{\beta \gamma}(\bar x,\bar v)=0$. Thus a geodesic $x\mapsto x(t)$  passing at $\check x$ with velocity $\check v$, satisfies for $(\check x,\check v)$ near $(\bar{x}, \bar{v})$
\[
\frac{\dd^2 x^ \alpha}{\dd t^2}=(\vert \check v-\bar v\vert^2+ \vert \check x-\bar x\vert) O(1),
\]
which means that the point particle approximately moves on a straight line (almost zero coordinate acceleration). However, there is an important difference with respect to the result for Lorentzian geometry. In Lorentzian geometry for $ \check x=\bar x$ we have the exact identity $\frac{\dd^2 x^ \alpha}{\dd t^2}=0$ at the point under consideration irrespective of the velocity, while in Finslerian theories this is not true, since $G^\alpha_{\beta \gamma}(\bar x, v)$ vanishes at $v=\bar{v}$ but depends on velocity. Thus:
\begin{quote}
{\em The observed trajectories passing through a chosen event ${\bar x}$ are `straight lines' (no coordinate acceleration) only approximately and only for particles moving slowly with respect to the observer}.
\end{quote}
We remark that observationally the difference with respect to Lorentzian theories cannot be appreciated, at least without some very fine measurements, indeed the faster the velocity with respect to the observer, the shorter the time that the particle will stay in a neighborhood of the observer.

We realize that the  experience according to which, free particles
move in straight lines always refers to particles which are slow with respect to the observer. Lorentz-Finsler geometry has therefore helped us to disclose a phenomenological extrapolation (particles locally move on straight lines) which does not really correspond to experience (without the slowness condition) and hence has helped us to ascertain that some mathematical restrictions can indeed be dropped.

We are going to improve the previous result. We need the notion of pullback connection. Given a local section $s\colon U\to TM\backslash 0$, $U\subset M$, we consider the pullback connection $\overset{s}{\nabla}:=s^*\nabla$ where $\nabla$ is the Finsler connection. This is an ordinary linear connection which has been studied extensively by Ingarden and Matsumoto \cite{ingarden93}, see also \cite[Sect.\ 4.1.1]{minguzzi14c}. The connection $s^*\nabla$ has connection coefficients \cite[Eq.\ (3.7)]{ingarden93} \cite[Sect.\ 4.1.1]{minguzzi14c}
\begin{equation}
\overset{s}{H}{}^\alpha_{\beta \gamma}(x)=H^\alpha_{\beta \gamma}\big(x,s(x)\big)+\left[\frac{\p s^\mu}{\p x^\gamma}+N^\mu_{\gamma}\big(x,s(x)\big)\right] V^\alpha_{\beta \mu}\big(x,s(x)\big)
\end{equation}
where $V^\alpha_{\beta \gamma}$ are the vertical connection coefficients. For the Berwald or Chern-Rund connections they vanish thus
\begin{equation} \label{jjd}
\overset{s}{H}{}^\alpha_{\beta \gamma}(x)=H^\alpha_{\beta \gamma}\big(x,s(x)\big).
\end{equation}
Ingarden and Matsumoto have calculated the torsion and curvature of the pullback connection. In the Berwald or Chern-Rund cases the torsion vanishes while the curvature is\footnote{This equation is easily obtained from the definition of  curvature of the linear connection \cite[Eq.\ (65)]{minguzzi14c}:  $R^\nabla(\check X,\check Y) \tilde{Z}= \nabla_{\check X}  \nabla_{\check Y} \tilde{Z}- \nabla_{\check Y}  \nabla_{\check X}\tilde{Z}-\nabla_{[\check X,\check Y]} \tilde{Z}$, using (notations of that work) $\overset{s}{R}=s^* R^\nabla$, and  $s_* X=D_Xs+\mathcal{N}(X)$ and analogously for $Y$. } \cite[Eq.\ (3.11)]{ingarden93}
\begin{equation} \label{dpo}
\overset{s}{R}{}^\alpha_{\ \beta \gamma \delta}=[R^{HH} {}^\alpha_{\ \beta \gamma \delta}+R^{VH} {}^\alpha_{\ \beta \mu \delta} D_\gamma s^\mu-R^{VH} {}^\alpha_{\ \beta  \mu \gamma} D_\delta s^\mu ] \vert_{v=s(x)} .
\end{equation}

It will turn out that it is particularly convenient to  use normal coordinates for a suitably chosen pullback connection. Indeed,
with these preliminaries we can improve  Prop.\ (\ref{poi})  as follows\footnote{Equation  (\ref{lagb}) has been previously obtained in \cite{pfeifer14}.}
\begin{proposition} \label{wfy}
Let $(\bar x,\bar v)\in E=TM\backslash 0$, and let us consider the Berwald or the Chern-Rund  Finsler connection. It is always possible to find local coordinates in a neighborhood of $\bar x$, for instance the $\bar v$-Berwald or the $\bar v$-Chern-Rund normal coordinates,
 such that $\p_0(\bar x)=\bar{v}$, $H^\alpha_{\beta \gamma}(\bar x,\bar v)=0$ and
\begin{align}
H^\alpha_{\beta \gamma,\delta}(\bar x,\bar v)&=-\frac{1}{3}\, \left(R^{HH} {}^\alpha_{\ \beta \gamma \delta}(\bar x,\bar v) + R^{HH} {}^\alpha_{ \ \gamma \beta \delta}(\bar x,\bar v)\right). \label{dpa}
\end{align}
If we are considering the Berwald connection, $H^\alpha_{\beta \gamma}=G^\alpha_{\beta \gamma}$, we can also conclude that in the new coordinates (e.g.\ in the $\bar v$-Berwald normal coordinates)
\begin{align}
g_{\alpha \beta,\gamma}(\bar x,\bar v)&=-2L_{\alpha \beta \gamma}(\bar x,\bar v), \qquad (\frac{\p \mathscr{L}}{\p v^\alpha})_{,\gamma}(\bar x,\bar v)= \mathscr{L}_{,\gamma}(\bar x,\bar v)=0 , \label{kkkb}\\
g_{\alpha \beta,\gamma,\delta}(\bar x,\bar v)&=-\frac{1}{6}\big(R_{\beta \gamma \alpha \delta}+R_{\alpha \gamma \beta \delta}+2R^\mu_{\ \gamma \nu \delta} v^\nu C_{\alpha \beta \mu}+6\nabla_\delta^{HB} L_{\alpha \beta \gamma} +\gamma/\delta\big)\vert_{(\bar{x},\bar{v})}, \label{metb}\\
(\frac{\p \mathscr{L}}{\p v^\alpha})_{,\gamma,\delta}(\bar x,\bar v)&=\frac{1}{3}\big( R_{\gamma \alpha \delta}\!+\! R_{\delta \alpha \gamma}\!-\!v^\beta \nabla^{HB}_\beta L_{\alpha \gamma \delta}+\!R^\mu_{\alpha} C_{\mu \gamma \delta}\!-\!R^\mu_{\gamma} C_{\mu \delta \alpha}\!-\!R^\mu_{\delta} C_{\mu \gamma \alpha} \big)\vert_{(\bar{x},\bar{v})} , \label{labb}\\
\mathscr{L}_{,\gamma,\delta}(\bar x,\bar v)&=\frac{1}{3}\,R_{\gamma \delta \alpha} \, v^\alpha \vert_{(\bar{x},\bar{v})}. \label{lagb}
\end{align}
where  $R_{\alpha \beta \gamma \delta}$ is the Berwald $HH$-curvature and
$\gamma/\delta$ means ``plus terms with  $\gamma$ and $\delta$ exchanged''.

If we are considering the Chern-Rund connection, $H^\alpha_{\beta \gamma}=\Gamma^\alpha_{\beta \gamma}$, we can also conclude that in the new coordinates (e.g.\ in the $\bar v$-Chern-Rund normal coordinates)
\begin{align}
g_{\alpha \beta,\gamma}(\bar x,\bar v)&=(\frac{\p \mathscr{L}}{\p v^\alpha})_{,\gamma}(\bar x,\bar v)= \mathscr{L}_{,\gamma}(\bar x,\bar v)=0 , \label{kkk}\\
g_{\alpha \beta,\gamma,\delta}(\bar x,\bar v)&=-\frac{1}{6}\big(R_{\beta \gamma \alpha \delta}+R_{\alpha \gamma \beta \delta}+2R^\mu_{\ \gamma \nu \delta} v^\nu C_{\alpha \beta \mu} +\gamma/\delta \big)\vert_{(\bar{x},\bar{v})}. \label{met}
\\
(\frac{\p \mathscr{L}}{\p v^\alpha})_{,\gamma,\delta}(\bar x,\bar v)&=\frac{1}{3}\big( R_{\gamma \alpha \delta}\!+\! R_{\delta \alpha \gamma}\!+\!R^\mu_{\alpha} C_{\mu \gamma \delta}\!-\!R^\mu_{\gamma} C_{\mu \delta \alpha}\!-\!R^\mu_{\delta} C_{\mu \gamma \alpha} \big)\vert_{(\bar{x},\bar{v})} , \label{lab}\\
\mathscr{L}_{,\gamma,\delta}(\bar x,\bar v)&=\frac{1}{3}\,R_{\gamma \delta \alpha} \, v^\alpha \vert_{(\bar{x},\bar{v})}. \label{lag}
\end{align}
where $R_{\alpha \beta \gamma \delta}$ is the Chern-Rund $HH$-curvature.
\end{proposition}

\begin{proof}
By Prop.\ (\ref{poi})  we can find a local system of coordinates such that $H^\alpha_{\beta \gamma}(\bar x,\bar v)=0$ and $\p_0= \bar v$. The section $s\colon U\to TM\backslash 0$ given in components by $s^\mu=\delta^\mu_0$ is such that $D_\delta s^\mu=0$ and $s(\bar x)=\bar v$. Observe that these properties do not depend on the coordinate system.

Thus, let us forget of the coordinate system constructed so far and let $s\colon U\to TM\backslash 0$, $\bar{x}\in U$, be a local section such that $s(\bar{x})=\bar v$, $D_\delta s^\mu=0$. We consider the torsionless pullback connection $s^*\nabla$. It is well known \cite{petrov69,eisenhart97,gray04} that given a linear  torsionless connection local coordinates exist which accomplish $\overset{s}{H} {}^\alpha_{\beta \gamma}(\bar{x})=0$, $\p_0=\bar v$ and
\[
\overset{s}{H} {}^\alpha_{\beta \gamma,\delta}(\bar{x})=-\frac{1}{3}\, \big(\overset{s}{R} {}^\alpha_{\ \beta \gamma \delta}(\bar x) + \overset{s}{R} {}^\alpha_{ \ \gamma \beta \delta}(\bar x)\big).
\]
From Eq.\ (\ref{jjd}) we have $H^\alpha_{\beta \gamma}(\bar x,\bar{v})=0$, which implies
$\frac{\p s^\mu}{ \p x^\delta}(\bar{x})=D_\delta s^\mu (\bar{x})=0$ since the components of the non-linear connection vanish at $\bar x$.
Thus  taking into account that
\begin{align*}
\overset{s}{H}{}^\alpha_{\beta \gamma, \delta}(\bar x)&=H^\alpha_{\beta \gamma,\delta}(x,\bar v)+ \frac{\p H^\alpha_{\beta \gamma}}{\p v^\mu} (x,\bar{v}) \, \frac{\p s^\mu}{ \p x^\delta}(\bar{x})=H^\alpha_{\beta \gamma,\delta}(x,\bar v),
\end{align*}
and Eq.\ (\ref{dpo}) we arrive at Eq.\ (\ref{dpa}).

If $\nabla$ is the Berwald connection, $H^\alpha_{\beta \gamma}=G^\alpha_{\beta \gamma}$, we have by definition of Landsberg tensor
\[
g_{\alpha \beta,\gamma}-2N^\mu_\gamma C_{\alpha \beta \mu}-g_{\mu \beta} G^\mu_{\alpha \gamma}-g_{\alpha \mu} G^\mu_{\beta \gamma}=-2 L_{\alpha \beta \gamma}.
\]
Equation (\ref{kkkb}) is an immediate consequence of this equation.
Differentiating with respect to $x^\delta$ and evaluating at $(\bar x,\bar v)$ where several coefficients vanish, we obtain
\[
(g_{\alpha\beta,\gamma,\delta}-g_{\mu \beta} G^\mu_{\alpha \gamma,\delta}-g_{\alpha \mu} G^\mu_{\beta \gamma ,\delta}-2 N^\mu_{\gamma,
\delta} C_{\alpha \beta \mu}+2 L_{\alpha \beta \gamma,\delta})\vert_{(\bar x,\bar v)}=0.
\]
Subtracting this equation with that obtained exchanging $\gamma$ and $\delta$ we obtain a known symmetry of the Berwald HH-curvature \cite[Eq.\ (86)]{minguzzi14c} \[
R^{HH}{}_{\beta \alpha \delta \gamma}+R^{HH}{}_{\alpha \beta \delta \gamma}+2R^\mu_{\delta \gamma}C_{\alpha \beta \mu}=2 (\nabla^{HB}_\delta L_{\alpha \beta \gamma}-\nabla^{HB}_{
\gamma}L_{\alpha \beta \delta}).
\]
Symmetrizing we  get instead Eq.\ (\ref{metb}). Contracting with $\bar v^\beta$, using \cite[Eqs.\ (70),(86)]{minguzzi14c} and using the second identity of \cite[Sect.\ 5.4.1]{minguzzi14c}  we get Eq.\ (\ref{labb}). Equation (\ref{lagb}) follows easily upon contraction with $\bar v^\alpha$.

If $\nabla$ is the Chern-Rund connection, $H^\alpha_{\beta \gamma}=\Gamma^\alpha_{\beta \gamma}$ we have
\[
g_{\alpha \beta,\gamma}-2N^\mu_\gamma C_{\alpha \beta \mu}-g_{\mu \beta} \Gamma^\mu_{\alpha \gamma}-g_{\alpha \mu} \Gamma^\mu_{\beta \gamma}=0.
\]
Equation (\ref{kkk}) is an immediate consequence of this equation.
Differentiating with respect to $x^\delta$ and evaluating at $(\bar x,\bar v)$ where several coefficients vanish, we obtain
\[
(g_{\alpha\beta,\gamma,\delta}-g_{\mu \beta} \Gamma^\mu_{\alpha \gamma,\delta}-g_{\alpha \mu} \Gamma^\mu_{\beta \gamma ,\delta}-2 N^\mu_{\gamma,
\delta} C_{\alpha \beta \mu})\vert_{(\bar x,\bar v)}=0.
\]
Subtracting this equation with that obtained exchanging $\gamma$ and $\delta$ we obtain a known symmetry of the Chern-Rund HH-curvature \cite[Eq.\ (87)]{minguzzi14c}
\begin{equation}
R_{\alpha \beta \gamma \delta}=-R_{\beta \alpha \gamma \delta}-2 R^\mu_{\gamma \delta} C_{\alpha \beta \mu}.
\end{equation}
Using Eq.\ (\ref{dpa}) and this identity
 we arrive at Eq.\ (\ref{met}). Contracting with $\bar v^\beta$ and using the second identity of \cite[Sect.\ 5.4.1]{minguzzi14c}  we get Eq.\ (\ref{lab}). Contracting instead with $\bar v^\alpha \bar v^\beta$ and using again \cite[Eq.\ (87)]{minguzzi14c} we get Eq.\ (\ref{lag}).

\end{proof}

\subsection{The observer and its adapted coordinates}
Let us assume that $g$ has Lorentzian signature.
Let $x\colon [0,1]\to M$ be a $C^2$ future-directed timelike curve parametrized with respect to proper time, namely such that $g_{\dot{x}}(\dot x,\dot x)=-1$.  We wish to construct adapted coordinate systems analogous to  Fermi-Walker's \cite{manasse63,ni78,li79,marzlin94,nesterov99}.
The acceleration of the curve is
\begin{equation} \label{qxy}
a=\tilde{D}_{\dot{x}} \dot{x}.
\end{equation}
Let $e_0=\dot x$  and let $\{e_i(t),i=1,\cdots, n\}$ be a $C^1$ $g_{\dot{x}}$-orthonormal base of the space orthogonal to $\dot{x}(t)$,  namely $\ker g_{\dot{x}}(\dot x, \cdot)(t)$.  The acceleration is orthogonal to the velocity since $0=\tilde D_{\dot x} g_{\dot x }(\dot x,\dot x)=2g_{\dot{x}}(\dot{x},a)$, thus we can write $a=a^\alpha e_\alpha$ for some components $\{a^\alpha\}$ with $a^0=0$.
 We have
\begin{equation} \label{kok}
\tilde D_{\dot{x}} e_i=\Omega_{ji}(t) e_{j}+a_i \dot{x}
\end{equation}
where $\Omega$ is antisymmetric. Indeed,
\[
  0=\tilde{D}_{\dot{x}} g_{\dot{x}}(\dot{x}, e_i)=g_{\dot{x}}(a, e_i)+g_{\dot{x}}(\dot{x}, \tilde{D}_{\dot{x}} e_i) ,
\]
which proves that Eq.\ (\ref{kok}) holds for some matrix $\Omega$, and
\[
  0=\tilde{D}_{\dot{x}} g_{\dot{x}}(e_i,e_j)=g_{\dot{x}}(\tilde{D}_{\dot{x}} e_i, e_j)+g_{\dot{x}}(\tilde{D}_{\dot x} e_i, e_j)=\Omega_{ji}+\Omega_{ij} ,
\]
which proves that $\Omega$ is antisymmetric. It can be written $\Omega_{ij}=-\epsilon_{ijk} \omega^k$ where $\omega=\omega^k{e_k}$ is the angular velocity of the frame.
Let us introduce an antisymmetric tensor defined over the curve through $\Omega^{\alpha \beta} e_\beta \otimes e_{\beta}$ where $\Omega^{i j}=\Omega_{ij}$, $\Omega^{0\alpha}=-\Omega^{\alpha0}=a^\alpha$. Let us lower the indices with $g_{\dot x}$ then
\begin{equation}
\tilde{D}_{\dot{x}} e_{\alpha}=\Omega_{\ \alpha}^{\beta} \, e_{\beta} ,
\end{equation}
and\footnote{The sign in the definition of $\Omega$, opposite to that in \cite{misner73}, is chosen so as to make Eq.\ (\ref{qki}) reminiscent of the classical equation for apparent forces.}
\begin{equation}
\Omega_{\alpha \beta}=\dot x_\alpha a_\beta -a_\alpha \dot x_\beta -\varepsilon_{\gamma \delta \alpha \beta} \dot x^\gamma \omega^\delta ,
\end{equation}
where $\varepsilon_{\alpha \beta \gamma \delta} =\sqrt{\vert g_{\dot{x}}\vert} \, [{\alpha \beta \gamma \delta} ]$ is the volume form.

A local laboratory can be represented through the base $\{e_0,e_i\}$ where $a$ and $\omega$ are the  acceleration and angular velocity of the laboratory as measured through dynamometers and gyroscopes from inside the laboratory. It is understood that a gyroscope with direction $e(t)$, $g_{\dot x}(\dot x, e)=0$, satisfies
\begin{equation}
\tilde D_{\dot{x}} e=g_{\dot x}(a, e) \dot{x},
\end{equation}
(a better motivation would pass from the study of extended bodies regarded as unions of point particles).
It is convenient to introduce a (Fermi-Walker) time derivative with respect to the observer as follows
\begin{equation}
\tilde D_{{\dot x}}^{FW} X=\tilde D_{{\dot x}} X-\Omega(X)
\end{equation}
where $\Omega(X)=\Omega^\alpha_{\ \beta} X^\beta e_{\alpha}$.
We remark that as seen from the observer the time derivatives of the acceleration and angular velocity are $\tilde D_{{\dot x}}^{FW} a=\dot{a}^i e_i$, and $\tilde D_{{\dot x}}^{FW} \omega=\dot \omega^i e_i$. By linearity the Fermi-Walker derivative extends to tensors and  it is easy to check that the derivative of the endomorphism $\Omega$ is $\tilde D_{{\dot x}}^{FW} \Omega=\tilde D_{{\dot x}}\Omega$. This observation will be relevant in Eq.\ (\ref{qki}).


\begin{proposition} \label{see}
Let $x:I\to M$ be a timelike curve parametrized with respect to proper time and let $\{\dot{x}, e_i\}$ be an orthonormal frame over the curve.
Coordinate systems $\{x^0=t,x^i\}$ such that $\dot{x}=\p_t$ and $e_i=\p_i$, exist. Moreover, let $\nabla$ be a Finsler connection (compatible with the non-linear connection of the spray) which is $HH$-torsionless, that is, such that the horizontal coefficients are symmetric $H^\alpha_{\beta \gamma}=H^\alpha_{\gamma \beta}$ and $V^\alpha_{\beta \gamma}=0$ (for instance the Berwald or the Chern-Rund connection). We have:
\begin{itemize}
 \item[(a) ]For any such coordinate system we have on the curve $H_{0\alpha}^\beta=\Omega_{\ \alpha}^{\beta}$, that is:
 \begin{align}
H^0_{00}(x(t),\dot x(t))&=0, \\
H^i_{0 0}(x(t),\!\dot x(t))=H^0_{i 0}(x(t),\!\dot x(t))&=a_i(t), \label{mss}\\
 H^j_{i 0}(x(t),\dot x(t))&=\Omega_{ji}(t). \label{nss}
\end{align}
\item[(b)]Some of these coordinate systems are also such that all the other components of $H$, namely $H^\alpha_{ij}(x,\dot x)$,  vanish over the curve.
\end{itemize}
\end{proposition}


\begin{proof}
The first claim is obvious. The coordinates system could be constructed introducing a Riemannian metric and using the exponential map of this metric from the curve, $f(t,{\bf x})=\exp^h_{x(t)}(x^i e_i)$,  so as to construct the coordinate system in a tubular neighborhood of the curve.
Statement (a) follows immediately from Eq.\ (\ref{qxy}) and (\ref{kok}), using $(e_\nu)^\mu=\delta^\mu_\nu$.

For (b)  let us start from a coordinate system as in (a). Over the mentioned coordinate neighborhood of $x$ let us consider a connection, namely a spray whose coefficients $C^\alpha_{\beta  \mu}$ are independent of velocity. Let  the connection be defined by $C^\alpha_{\beta  \mu}(t,{\bf x})=H^\alpha_{\beta  \mu}(x(t),\dot x(t))$, hence independent of ${\bf x}$.

Through the exponential map determined by $C$, $x(t,{\bf x})=\exp^C_{x(t)}(\tilde x^i e_i)$, we can define coordinates $(t,\tilde {\bf x})$ on a neighborhood of the curve. In the new coordinate system the coefficients of the connection are denoted by $\tilde C^\alpha_{\beta  \mu}$ while those of the spray by $\tilde H^\alpha_{\beta  \mu}$. We have
\[
\tilde C^\alpha_{\beta \gamma}(x) =\frac{\p \tilde x^\alpha}{\p x^\sigma } \frac{\p x^\delta}{\p \tilde x^\beta} \frac{\p x^\mu}{\p \tilde x^\gamma} \, C^\sigma_{\delta \mu}(x) +\frac{\p^2  x^\sigma}{\p \tilde x^\beta  \p  \tilde x^\gamma }  \frac{\p \tilde x^\alpha}{\p  x^\sigma},
\]
and analogously with $C(x)$ replaced by $H(x,v)$.
Observe that the new coordinate system is still such that $\p_0=e_0$, $\p_i=e_i$, thus the connection coefficients $H^\beta_{0\alpha}$ mentioned in (a) remain the same. The geodesic $\sigma(s)$ issued from $x(t)$ with direction $n^i e_i$ has equation $\tilde{x}^i=n^i s$, $t=cnst$, thus the geodesic condition $\nabla^C_{\sigma'} \sigma'=0$ at $x(t)$ reads
$\tilde C^\alpha_{ij}(x(t)) n^i n^j=0$, which due to the arbitrariness of $n$ implies $\tilde C^\alpha_{ij}(x(t))=0$. Thus the coordinate change sends $C^\alpha_{\beta \gamma}$ to $\tilde C^\alpha_{\beta \gamma}$, where the latter is such that $\tilde C^\alpha_{ij}=0$ over the curve. But since $H^\alpha_{\beta \gamma}$ is sent to $\tilde H^\alpha_{\beta \gamma}$ via the same transformation rule, and $H^\alpha_{\beta \gamma}=C^\alpha_{\beta \gamma}$ on the curve we can conclude that $\tilde H^\alpha_{i j}(x(t),\dot x(t))=0$.
\end{proof}
%

Proposition \ref{see} can be improved as follows

\begin{proposition} \label{dib}
With the assumptions of Prop.\ (\ref{see}) there are coordinate systems such that the coefficients $H^\alpha_{\beta \gamma}$ vanish over the curve saved for
\begin{align}
H^i_{0 0}(x(t),\dot x(t))=H^0_{i 0}(x(t),\dot x(t))&=a_i(t), \label{ngs}\\
 H^j_{i 0}(x(t),\dot x(t))&=\Omega_{ji}(t), \label{bgs}
\end{align}
Among these coordinate systems  those obtained from the Fermi construction for the pullback connection are such that over the curve (i.e.\ at $(x(t),\dot x(t))$)
\begin{align}
H^0_{00,0}&=H^\alpha_{i j,0}=0 ,\\
H^i_{0 0,0}&=H^0_{i 0,0}=\dot a_i, \\
 H^j_{i 0,0}&=\dot \Omega_{ji},\\
 H^0_{i 0,j}&=-a_i(t) a_j(t)+R^{HH} {}^0_{\ i j 0}+R^{VH} {}^0_{\ i m 0} \Omega_{m j}-R^{VH} {}^0_{\ i m j} a_m, \label{tse}\\
 H^k_{i 0,j}&=\Omega_{j k}(t) a_i(t)+R^{HH} {}^k_{\ i j 0}+R^{VH} {}^k_{\ i m 0} \Omega_{m j}-R^{VH} {}^k_{\ i m j} a_m, \label{tsf}\\
H^0_{0 0,i}&=\dot a_i-\Omega_{ik}(t) a_k(t)+R^{VH} {}^0_{\ 0 m 0} \Omega_{ m i}-R^{VH} {}^0_{\ 0 m i} a_m ,\\
H^j_{0 0,i}&=\dot \Omega_{ji}+\Omega_{ik}(t)\Omega_{kj}(t)+a_i(t)a_j(t)+R^{HH} {}^j_{\ 0 i 0}\nonumber\\&\qquad  \qquad+R^{VH} {}^j_{\ 0 m 0} \Omega_{m i}-R^{VH} {}^j_{\ 0 m i} a_m ,\\
H^\alpha_{j k,l}&=-\frac{1}{3}\, \big(R^{HH} {}^\alpha_{\ j k l }  +R^{VH} {}^\alpha_{\ j m l} \Omega_{m k}  -R^{VH} {}^\alpha_{\ j m k} \Omega_{ m l} + k/j \big), \label{tsd}
\end{align}
where
$k/j$ means ``plus terms with  $k$ and $j$ exchanged''.  From here several other equations are easily obtained, for instance in the geodesic ($a^i=0$) parallel transport ($\Omega_{ij}=0$), Chern-Rund connection case we obtain
\begin{align}
g_{00,i,j}&=-R^{HH}{}_{0i0j}-R^{HH}{}_{0j0i},\\
g_{i0,j,k}&=R^{HH}{}_{ijk0}-\frac{1}{3} (R^{HH}{}_{0ijk}+R^{HH}{}_{0jik}),\\
g_{ij,k,l}&=-\frac{1}{3}\, \big(R^{HH}{}_{ijkl}+R^{HH}{}_{ikjl}+R^{HH}{}_{jikl}+R^{HH}{}_{jkil}\big) \nonumber\\
\qquad &+C^m_{ij}\big(R_{lmk}+R_{kml} +2R^s_m C_{skl}-2R^s_k C_{slm}-2R^s_l C_{skm} \big) ,
\end{align}
while all the other derivatives with respect to position of first and second order vanish.

\end{proposition}

Once a Finsler connection has been chosen some simplifications are possible in Eq.\ (\ref{tse})-(\ref{tsd}). For instance, in the Berwald case $R^{VH}$ is the Berwald curvature, $R^{VH} {}^\alpha_{\ \beta \gamma \delta}=G^\alpha_{\ \beta \gamma \delta}$, which vanishes whenever one of the lower indices is zero, while $-\frac{1}{2}R^{VH} {}^0_{\ \alpha \beta \gamma}$ is the Landsberg tensor \cite{minguzzi14c}.

\begin{proof}
By Prop.\ \ref{see} there is a coordinate system such that $\dot{x}=\p_0$, $e_i=\p_i$, and (a) holds. Thus the section given by $s^\mu=\delta^{\mu}_0$ is such that
\begin{align*}
D s&=N^\mu_\alpha(x, \dot{x})e_{\mu}\otimes \dd x^\alpha=H^\mu_{0 \alpha}(x, \dot{x}) e_\mu\otimes \dd x^\alpha\\
&=a_i (e_i\otimes \dd t+\dot{x} \otimes \dd x^i)+\Omega_{j i} (e_j\otimes \dd x^i).
\end{align*}

Let $s\colon U \to M$ be a section defined in a neighborhood $U$ of the curve such that the previous identity holds on the curve. We consider the  pullback connection $s^*\nabla$ where $\nabla$ is a HH-torsionless Finsler connection such that $V^\alpha_{\beta \gamma}=0$. As a consequence $s^*\nabla$ is torsionless and Eqs.\ (\ref{jjd}) and (\ref{dpo}) hold. We know that this connection satisfies (a) of Prop.\ (\ref{see}) in a coordinate system for which $\dot{x}=\p_0$, $e_i=\p_i$ thus in the same coordinate system the pullback connection is such that
\begin{align}
\overset{s}{H}{}^0_{00}(x(t))&=0, \label{b0} \\
\overset{s}{H}{}^i_{0 0}(x(t))=\overset{s}{H}{}^0_{i 0}(x(t))&=a_i(t), \label{b1}\\
 \overset{s}{H}{}^j_{i 0}(x(t))&=\Omega_{ji}(t). \label{b2}
\end{align}
Now, of all the coordinate systems for which $\dot{x}=\p_0$, $e_i=\p_i$, we choose the one for which the coordinates are constructed with the Fermi prescription, namely through the exponential map $f(t,{\bf x})=\exp^s_{x(t)} (x^i e_i)$. Since the curves of equation $x^i=n^i u$ are geodesics for the pullback connection we have $\overset{s}{H}{}^\alpha_{i j}((t,nu)) n^i n^j=0$, which by the arbitrariness of $n$ gives (for the latter differentiate first with respect to $u$ and then set $u=0$)
\begin{align*}
\overset{s}{H}{}^\alpha_{i j}(x(t))&=0,\\
\overset{s}{H}{}^\alpha_{i j,k}(x(t))+\overset{s}{H}{}^\alpha_{k i ,j}(x(t))+\overset{s}{H}{}^\alpha_{j k,i}(x(t))&=0 .
\end{align*}
Differentiating (\ref{b0})-(\ref{b2}) we also get
\begin{align}
\overset{s}{H}{}^0_{00,0}(x(t))=\overset{s}{H}{}^\alpha_{i j,0}(x(t))&=0 ,\\
\overset{s}{H}{}^i_{0 0,0}(x(t))=\overset{s}{H}{}^0_{i 0,0}(x(t))&=\dot a_i(t), \\
 \overset{s}{H}{}^j_{i 0,0}(x(t))&=\dot \Omega_{ji}(t).
\end{align}
From here using
\[
\overset{s}{R}{}^{\alpha}_{\ \beta \gamma \delta}=\overset{s}{H}{}^\alpha_{\beta \delta,\gamma}-\overset{s}{H}{}^\alpha_{\beta \gamma,\delta}+\overset{s}{H}{}^\alpha_{\mu \gamma} \overset{s}{H}{}^\mu_{\beta \delta}-\overset{s}{H}{}^\alpha_{\mu \delta}\overset{s}{H}{}^\mu_{\beta \gamma} ,
\]
we obtain
\begin{align}
\overset{s}{H} {}^0_{0 i,j}(x(t))&=-a_i(t) a_j(t)+\overset{s}{R} {}^0_{\ i j 0}(x(t))  , \label{yhb} \\
\overset{s}{H} {}^k_{i 0,j}(x(t))&=\Omega_{jk}(t) a_i(t)+\overset{s}{R} {}^k_{\ i j 0}(x(t)) ,\label{yhc} \\
\overset{s}{H} {}^0_{0 0,i}(x(t))&=\dot a_i(t)-\Omega_{ik}(t) a_k(t)+\overset{s}{R} {}^0_{\ 0 i  0}(x(t)) ,\\
\overset{s}{H} {}^j_{0 0,i}(x(t))&=\dot \Omega_{ji}(t)+\Omega_{ik}(t)\Omega_{kj}(t)+a_i(t)a_j(t)+\overset{s}{R} {}^j_{\ 0 i 0}(x(t)) ,\\
\overset{s}{H} {}^\alpha_{j k,l}(x(t))&=-\frac{1}{3}\, \big(\overset{s}{R} {}^\alpha_{\ j k l}(x(t)) + \overset{s}{R} {}^\alpha_{ \ k j l}(x(t))\big). \label{ygb}
\end{align}
We use
\begin{align*}
\overset{s}{H}{}^\alpha_{\beta \gamma, \delta}(x(t))&=H^\alpha_{\beta \gamma,\delta}(x(t),\dot x(t))+ \frac{\p H^\alpha_{\beta \gamma}}{\p v^\mu} (x(t),\dot x(t)) \, \frac{\p s^\mu}{ \p x^\delta}(\bar{x})\\
&=H^\alpha_{\beta \gamma,\delta}(x(t),\dot x(t))+ \frac{\p H^\alpha_{\beta \gamma}}{\p v^\mu} (x(t),\dot x(t)) \, [D_\delta s^\mu(x(t)) - H^\mu_{0\delta}(x(t),\dot x(t))]\\
&=H^\alpha_{\beta \gamma,\delta}(x(t),\dot x(t))
\end{align*}
 where in the last step we observed that  $H^\mu_{0\delta}(x(t),\dot x(t))=\overset{s}{H} {}^\mu_{0\delta}(x(t))$  and used Eqs.\ (\ref{b1}) and (\ref{b2}).

Now we recall that in general from \cite[Sect.\ 5.2.2]{minguzzi14c} $R^{VH}{}^\alpha_{\ \beta \gamma \delta}(x,v) v^\gamma=0$,  and use Eq.\ (\ref{dpo}) to get
\begin{align*}
\overset{s}{R} {}^\alpha_{\ i j 0} &=R^{HH} {}^\alpha_{\ i j 0}+R^{VH} {}^\alpha_{\ i m 0} \Omega_{mj}-R^{VH} {}^\alpha_{\ i m j} a_m,\\
\overset{s}{R} {}^0_{\ 0 i 0} &=R^{HH} {}^0_{\ 0 i 0}+R^{VH} {}^0_{\ 0 m 0} \Omega_{ mi}-R^{VH} {}^0_{\ 0 m i} a_m, \\
\overset{s}{R} {}^j_{\ 0 i 0} &=R^{HH} {}^j_{\ 0 i 0}+R^{VH} {}^j_{\ 0 m 0} \Omega_{ mi}-R^{VH} {}^j_{\ 0 m i} a_m, \\
\overset{s}{R} {}^\alpha_{\ j k l} &=R^{HH} {}^\alpha_{\ j k l }+R^{VH} {}^\alpha_{\ j m l} \Omega_{ mk}-R^{VH} {}^\alpha_{\ j m k} \Omega_{ ml}.
\end{align*}
Observe that we can further use $R^{HH} {}^0_{0 \alpha \beta}=R^0_{\alpha \beta}=0$ in the second equation (contract the second equation in display in \cite[Sect.\ 5.4.1]{minguzzi14c} with $y^i y^j$). From these equations the thesis follows easily upon substitution in Eqs.\ (\ref{yhb})-(\ref{ygb}).
\end{proof}

\section{The free particle seen by the observer}

Let us consider a timelike geodesic $y(t)$ and a timelike curve $x(t)$ both parametrized with respect to proper time. Let us suppose that $y(t)$ remains close to $x(t)$ in both position and velocity for some proper time interval. Let $e_i(t)$ be a frame orthogonal to $\dot x(t)$. We can use the special coordinate system $\{x^\alpha\}$ constructed in the previous section to express the geodesic equation
\[
\ddot y^\alpha+2G^\alpha(y,\dot y)=0,
\]
where $y^\alpha(t)=x^\alpha(y(t))$. We can Taylor expand the second term at $(x(t), \dot x(t))=((t,{\bf 0}),(1,{\bf 0}))$ setting $\xi= y-x$. We retain only the linear terms since terms of  higher order are unlikely to be observable
\begin{align*}
2G^\alpha(y,\dot y)&=2G^\alpha(x,\dot x)+2G^\alpha_\beta(x,\dot x) \dot \xi^\beta+2 G^\alpha_{,\beta}(x,\dot x) \xi^\beta +\cdots \\
&= G^\alpha_{0 0}(x,\dot x)+2G^\alpha_{\beta 0}(x,\dot x) \dot \xi^\beta+  G^\alpha_{00,\beta}(x,\dot x) \xi^\beta  +\cdots\\
&= H^\alpha_{0 0}(x,\dot x)+2H^\alpha_{\beta 0}(x,\dot x) \dot \xi^\beta+  H^\alpha_{00,\beta}(x,\dot x) \xi^\beta  +\cdots
\end{align*}
In the last line we have simply observed that the second line can be expressed entirely in terms of the components of the non-linear connection $N^\alpha_\beta=G^\alpha_{\beta 0}$ and that they can be expressed through the horizontal coefficients of any notable Finsler connection:   $N^\alpha_\beta=H^\alpha_{\beta 0}$.

The geodesic equation becomes a system
\begin{align}
0&=\frac{\dd^2 \xi^0}{\dd t^2}+2 a_i \dot \xi^i+\left( \dot a_i-\Omega_{ik} a_k +R^{VH} {}^0_{\ 0m0} \Omega_{mi}-R^{VH} {}^0_{\ 0mi} a_m\right)\xi^i,\\
0&=\frac{\dd^2 \xi^i}{\dd t^2}+a^i+2\Omega_{ij} \dot \xi^j+ R^i_j \xi^j+\big( \dot{\Omega}_{ij}+\Omega_{jk} \Omega_{ki}+a_i a_j \nonumber \\ & \qquad  + R^{VH} {}^i_{\ 0m0} \Omega_{mj}-R^{VH} {}^i_{ \ 0mj} a_m \big) \xi^j+\dot{a}^i \xi^0+2a^i \dot \xi^0 . \label{lzz}
\end{align}
Some comments are in order. The parameter $t$ appearing in these equations is the proper time parametrization of the curves and over $y(t)$ should be distinguished from the coordinate $x^0(y(t))=y^0(t)$ (recall that $\xi^0=y^0(t)-t$, hence $\dot y^0=1+\dot \xi^0$). If one is not really interested in the proper time parametrization of the geodesic but just on its spaceetime trajectory then it is natural to parametrize it with the local time foliation constructed by the observer, namely one can make a change of parameter in the second equation recasting it as a differential equation for $\xi(t(y^0))$. Observe that with this purpose in mind one could make at any chosen instant a change affine parameter over $y$, $t \to at+b$, so as to obtain $\xi^0=0$, $\dot \xi^0=0$. This operation makes the parametrizations of affine parameter and that of the foliation locally coincident, the error being of higher order than linear. This operation clearly removes the last relativistic terms of the second equation at least for some time interval.

One can proceed  in a different way using the Fermi-Walker derivative. It must be recalled that (\ref{lzz})  depends on the adapted coordinate system, which, as mentioned, depends on the chosen Finsler connection. Indeed,
as the construction of the coordinate system makes $H^\alpha_{ij}$ vanish, the coordinate system is different depending on whether $H^\alpha_{\beta \gamma}$ are the horizontal coefficients of the Berwald connection or of the Chern-Rund connection. In the former case  $R^{VH} {}^\alpha_{\ 0 \beta \gamma}=0$ while in the latter case $R^{VH} {}^\alpha_{\ 0 \beta \gamma}=L^\alpha_{\beta \gamma}$, which is the Landsberg tensor \cite[Sect.\ 5.3.2]{minguzzi14c}. We rewrite (\ref{lzz}) in the Chern-Rund case as the space part of
\begin{equation} \label{qki}
\tilde D_{\dot x}^{FW} \tilde D_{\dot x}^{FW} \xi=\underbrace{-R(\xi,\dot x)}_\text{tidal}\!\!\!\underbrace{-a}_\text{translat.}\!\!\underbrace{-(\tilde D_{\dot x}^{FW} \Omega)(\xi)}_\text{azimuthal}\underbrace{-\Omega(\Omega(\xi))}_\text{centrifugal}\underbrace{-2\Omega(\tilde D_{\dot x}^{FW} \xi)}_\text{Coriolis}\underbrace{+L(a,\xi)}_\text{Finslerian}
\end{equation}
where the first tidal term involves the non-linear curvature. In the Berwald case the last term does not appear, thus this Finslerian term depends on how we extend the coordinate frame.

We studied the geodesic equation in the non-inertial frame precisely because otherwise no Finslerian term appears at the linear order. In any case the weak equivalence principle is satisfied as previously discussed: slow particles with respect to the free falling observer move approximately over straight lines. It can be mentioned that in the Lorentzian case there appeared studies of the deviation equation which drop the condition on the slowness of the particle \cite{perlick08}.

For dimensional reasons $L$ should be an inverse length, thus the Finslerian contribution could be observed only if this length is not too large compared with the distance among the curves.
Elsewhere \cite{minguzzi14c} I have suggested that our spacetime could be  Landsbergian ($L=0$)  in which case the additional term would vanish even in the Chern-Rund case.

\section{Conclusions}

We have shown that special local coordinate systems can be constructed which have several properties in common with normal or Fermi(-Walker) coordinates. They simplify considerably the expression of the horizontal connection coefficients allowing their expression in terms of curvature invariants. In short our strategy applied the usual normal or Fermi coordinates construction to a pullback connection obtained from a suitable section $s\colon M\to TM$. Then we used some results by Ingarden and Matsumoto in order to relate the curvature of the pullback connection with the $HH$-curvature of the original Finsler connection. In the introduction we argued that for some choices of  section the Douglas-Thomas normal coordinates are recovered, and so the found expressions for the derivatives of the metric or for the connection coefficients hold for these coordinate systems as well.

Although the section  $s$ used in the construction privileges some vector, $s(\bar x)=\bar v$, the whole procedure is quite natural particularly for Fermi coordinates since there we have already a privileged vector given by the tangent of the curve. 

We also applied these findings to the study of the equivalence principle.
We have been able to write the geodesic equation for neighboring free falling particles in adapted coordinates, and in fact to separate the contributions from various terms, see Eq.\ (\ref{qki}), identifying one term of Finslerian origin. This term is related to the Landsberg tensor and is obtained whenever the local coordinate system is constructed so as to make the Chern-Rund connection coefficients vanish (as far as possible). It turns out that free particles would appear  as moving uniformly on straight lines, at least approximately, provided they move slowly with respect to the observer. On the contrary, in Lorentzian theories a free falling particle does not have coordinate acceleration at the location of the observer irrespective of the magnitude of its velocity.

\section*{Acknowledgments} Work partially supported by GNFM of INDAM.

\end{document}